%% file: master_article3.tex
\documentclass[a4paper, abstract=true, 10pt, DIV=10]{scrartcl}

\usepackage{etex}
\usepackage[utf8]{inputenc}
\usepackage[T1]{fontenc}
\usepackage{lmodern}
\usepackage[english]{babel}
\usepackage{textcomp}
\usepackage{enumitem}
\usepackage{microtype}
\usepackage{graphicx}
\usepackage{array}

\usepackage{amsmath}
\usepackage{amssymb}
\usepackage{amsthm}
\usepackage{amstext}

\usepackage{tikz}
\usetikzlibrary{math}
\usetikzlibrary{calc}
\usetikzlibrary{arrows.meta}
\usetikzlibrary{backgrounds}

\theoremstyle{plain}
\newtheorem{lemma}{Lemma}[section]

\newtheorem{theorem}[lemma]{Theorem}

\theoremstyle{definition}
\newtheorem{definition}[lemma]{Definition}
\theoremstyle{remark}

\usepackage[linesnumbered, noline, ruled, longend]{algorithm2e}

\usepackage[labelformat=simple]{subcaption}

\newcommand{\blocking}[1][\delta]{\text{Blocking}_{#1}}

\begin{document}

\tikzstyle{vertex} = [circle, draw, fill=black, inner sep=0pt, minimum width=3pt]
\tikzstyle{boundary} = [circle, draw, fill=white, inner sep=0pt, minimum width=3.5pt]
\tikzstyle{weight} = [font=\small, auto]
\tikzstyle{continued} = [loosely dotted, thick]
\tikzstyle{multiple} = [dashed]
\tikzstyle{arrow} = [-{Stealth[length=2mm]}]

\title{Online Graph Exploration on Trees, Unicyclic Graphs and Cactus Graphs}
\author{Robin Fritsch\thanks{\textsc{Technical University of Munich}, \emph{Email:} \texttt{robin.fritsch@tum.de}}}
\date{}

\maketitle

\begin{abstract}
We study the problem of exploring all vertices of an undirected weighted graph that is initially unknown to the searcher.
An edge of the graph is only revealed when the searcher visits one of its endpoints.
Beginning at some start node, the searcher's goal is to visit every vertex of the graph before returning to the start node on a tour as short as possible.

We prove that the Nearest Neighbor algorithm's competitive ratio on trees with $n$ vertices is $\Theta(\log n)$, i.e. no better than on general graphs.
Furthermore, we examine the algorithm Blocking for a range of parameters not considered previously and prove it is 3-competitive on unicyclic graphs as well as $5/2+\sqrt{2}\approx 3.91$-competitive on cactus graphs.
The best known lower bound for these two graph classes is 2.
\end{abstract}

\section{Introduction}
Exploration and map construction problems arise in robotics when a robot is tasked with exploring an unknown environment \cite{berman}.
While moving around, the robot gathers information about its surroundings and with this must decide how to proceed in its exploration.
This problem can be modeled as exploring an unknown graph.

We consider the fixed graph scenario, first introduced in \cite{kaly}, in which a connected undirected graph $G=(V,E)$ with $n=|V|$ vertices is explored.
Each edge $e\in E$ has a positive weight $|e|$ and the graph contains a distinguished start node $s\in V$ from which the searcher begins its exploration.
We assume that each vertex has an assigned unique identifier (ID).
Upon arriving at a vertex for the first time, the searcher obtains the weights of all edges incident to that vertex as well as the IDs of all adjacent vertices.
The searcher must visit every vertex of the graph before finally returning to the start node.

To measure the performance of an online algorithm, we use competitive analysis which compares its solution to the solution of the corresponding offline problem, in our case the traveling salesperson problem.
We call an online exploration algorithm \emph{$c$-competitive} if it produces a tour no longer than $c$ times the optimal (offline) tour for every instance.
The \emph{competitive ratio of} an online algorithm is defined as the infimum over all $c$ such that the algorithm is $c$-competitive.

The best known algorithms on general graphs are Nearest Neighbor (NN) \cite{rosenkrantz} and hierachical DFS \cite{megow} both with a competitive ratio of $\Theta(\log n)$.
For NN this worst-case ratio is tight even on planar unit-weight (unweighted) graphs \cite{hurkens}.
In particular, no algorithm with constant competitive ratio is known on general graphs.
On the other hand, the best known lower bound on the competitive ratio of an online algorithm has recently been improved from 2.5 \cite{dobrev} to $10/3$ \cite{birx}.

Algorithms with constant competitive ratio are known for several restricted graph classes.
The Blocking algorithm, introduced as \emph{ShortCut} \cite{kaly} and reformulated due to a "precarious issue in the formal implementation" \cite{megow}, is known to be 16-competitive on planar graphs.
Megow et al.\ \cite{megow} also showed that Blocking has a constant competitive ratio on graphs with bounded genus but not in general.
Furthermore, they present a hierarchical generalization of DFS which is $2k$-competitive on graphs with at most $k$ distinct weights and use this to construct an algorithm that is $\Theta(\log n)$-competitive on general graphs.

The problem was solved on cycles by Miyazaki et al.\ \cite{miyazaki} who found an algorithm that achieves the best possible competitive ratio on that graph class of $(1+\sqrt{3})/2\approx 1.366$.
Furthermore, they proved a general lower bound of 2 on unweighted graphs which DFS achieves on this graph class.
Brandt et al.\ \cite{brandt} studied tadpole graphs (a cycle with a path attached to it) and showed that NN is $2$-competitive on them.
Additionally, they extended Miyazaki's lower bound example for unweighted graphs \cite{miyazaki} to tadpole graphs to prove this is also the best achievable competitive ratio.

\subsection{Our Results}
We prove that the tight lower bound of $\Theta(\log n)$ on the competitive ratio of NN also holds on trees which improves the previous lower bound of $\Theta(\log n/\log\log n)$ \cite{wattenhofer}.
We do so by modifying a graph construction Hurkens and Woeginger \cite{hurkens} use to prove the lower bound on planar unit-weight graphs.

Furthermore, we prove upper bounds for two more graph classes: 
For unicyclic graphs, i.e. graphs that contain exactly one cycle, we prove that Blocking is 3-competitive.
We achieve this by examining the algorithm for a range of parameters which had not been considered previously.
This way, we also prove it to be $5/2+\sqrt{2}\approx 3.91$-competitive on cactus graphs, i.e.\ graphs in which any two cycles have at most one vertex in common.
The best known lower bound for these two graph classes is 2.

\subsection{Further Related Work}
Some of the first formal models for exploration problems were introduced by Papadimitriou and Yannakakis in \cite{papa} in which they search for the shortest path between two points in an unknown environment.
Following this, the problem of exploring an unknown graph was studied by Deng and Papadimitriou \cite{deng}, Albers and Henzinger \cite{albers} and Fleischer and Trippen \cite{fleischer}.
They worked with strongly connected directed graphs and the premise that all vertices as well as all edges of the graph are required to be explored.
Our setting in which only all vertices have to be visited has also been studied on directed graphs by Förster and Wattenhofer \cite{foerster} who proved upper and lower bounds on the best-possible competitive ratio linear in the number of vertices.

\section{Nearest Neighbor on Trees}
On general graphs, the Nearest Neighbor algorithm is $\Theta(\log n)$-competitive \cite{rosenkrantz}. We will prove that this bound is tight even for the simple graph class of trees.
To do so, we modify a construction Hurkens and Woeginger \cite{hurkens} use to prove the tightness of the bound for unweighted planar graphs.
Instead of considering a path of triangles, as they do, we construct a path with unit length edges (spikes) attached to it.

\input{figures/trees_nn_elsevier}

We recursively define graphs $G_k$ for $k\geq 1$ containing three distinguished vertices $l_k, r_k$ and $m_k$ (see Figure \ref{fig:trees_nn_G_k}).
The graph $G_1$ simply consists of the two unit length edges $l_1r_1$ and $r_1m_1$.
For $k\geq 2$, we construct $G_k$ by placing two copies $G_{k-1}'$ and $G_{k-1}''$ of $G_{k-1}$ next to each other and, in the middle, adding a new vertex $m_k$.
To connect the components, we add an edge of length $k$ between $r_{k-1}'$ and $l_{k-1}''$ as well as a unit weight edge between $l_{k-1}''$ and $m_k$.
Finally, we set $l_k=l_{k-1}'$ and $r_k = r_{k-1}''$.
Let $p_k$ be the length of the shortest path from $l_k$ to $r_k$ in $G_k$.
Since $p_1=1$ and $p_k = 2p_{k-1}+k$ for $k\geq 2$, a simple induction shows that $p_{k}=2^{k+1}-k-2$.

\begin{lemma}\label{lem:G_k_subgraph}
For $k\geq 1$, consider a graph $G$ that contains $G_k$ as a subgraph.
Furthermore, assume that edges between $G_k$ and $G-G_k$ are either incident to $l_k$ and have a length of at least 1 or are incident to $r_k$ and have a length of at least $k+1$.

Then there exists a partial NN tour exploring all of $G_{k}$ that starts in $l_k$, finishes in $m_k$ and has a length of $(k+1)2^k - 2$.
\end{lemma}

\begin{proof}
We prove this by induction on $k$.
For $k=1$ the tour from $l_1$ to $r_1$ to $m_1$ satisfies all conditions.
For $k\geq 2$, assume that NN resides in $l_{k}=l_{k-1}'$.
Note that both $G_{k-1}'$ and $G_{k-1}''$ as subgraphs of $G$ satisfy all conditions from the lemma.
Therefore, by the induction hypothesis, NN may next explore all of $G_{k-1}'$ on a tour of length $k 2^{k-1}-2$ ending up in $m_{k-1}'$.
From $m_{k-1}'$, the shortest path to any vertex of $G-G_k$ includes $l_k=l_{k-1}'$ and has a length of at least $1+(k-1)+p_{k-2}+1$.
Since the distance to $l_{k-1}''$ is $1+p_{k-2}+k$, NN may visit $l_{k-1}''$ next. (This argument also holds for $k=2$ with $p_0=0$ since we can already start the recursive construction at $G_0$ consisting of a single vertex $l_0=r_0=m_0$.)
Again by the induction hypothesis, NN may next explore all of $G_{k-1}''$ finishing in $m_{k-1}''$.
At this point, the shortest path from $m_{k-1}''$ to any vertex of $G-G_k$ includes either $r_k=r_{k-1}''$ or $l_k=l_{k-1}'$.
In the former case, the shortest path has a length of at least $1+p_{k-2}+(k+1)$.
In the latter case, the length of the shortest path is at least $3k+3p_{k-2}>k+2+p_{k-2}$.
With $m_k$ only being at a distance of $1+(k-1)+p_{k-2}+1$, NN will visit this vertex next.
The total length of the tour taken by NN from $l_k$ to $m_k$ is
\begin{equation*}
    2 \left(k 2^{k-1}-2\right) + 2\left( p_{k-2}+k+1\right) = (k+1)2^k - 2.
\end{equation*}
\end{proof}

\begin{theorem}\label{thm:trees_nn}
The competitive ratio of NN on trees is $\Theta(\log n)$.
\end{theorem}
\begin{proof}
The upper bound follows directly from the general case \cite{rosenkrantz}.
For the lower bound, consider $G_k$ and let $l_k$ be the start node.
Then NN explores $G_k$ on the tour described in Lemma \ref{lem:G_k_subgraph} and, finally, returns from $m_k$ to $l_k$.
Hence,
\begin{equation*}
    \mathrm{NN}(G_k)= (k+1)2^k - 2 + 1 + k + p_{k-1} = (k+2)2^k -2
\end{equation*}
On the other hand, let $w_k$ be the total weight of $G_k$.
Then $\mathrm{OPT}(G_k) = 2w_k = 6\cdot 2^k -2k-6$, which follows from $w_1=2$ and $w_{k}=2w_{k-1}+k+1$ for $k\geq 2$.
Finally, it is easy to see inductively that $G_k$ has $n=2^{k+1}-1$ vertices which implies $k+1=\log(n+1)$.
This proves
\begin{equation*}
    \frac{\mathrm{NN}(G_k)}{\mathrm{OPT}(G_k)}=\frac{(k+2)2^k-2}{6\cdot 2^k -2k-6} \geq \frac{k+2}{6} = \frac{\log(n+1)+1}{6}.
\end{equation*}
\end{proof}

\section{Blocking on Unicyclic Graphs and Cactus Graphs}

The algorithm $\blocking$ is a generalization of DFS \cite{megow}. It uses a \emph{blocking condition} which, depending on a fixed \emph{blocking parameter $\delta \in \mathbb{R}$}, determines when to delay the traversal of an edge, possibly forever.

\begin{definition}[Boundary edge]
During the exploration, we call an edge a \emph{boundary edge} when one of its endpoints has been visited while the other has not.
\end{definition}

Whenever we define a boundary edge in the form $e=(u,v)$, the first vertex (in this case $u$) has been visited while the second has not.

\begin{definition}[Blocking condition]
A boundary edge $e=(u,v)$ is \emph{blocked} by another boundary edge $e'=(u',v')$ if $e'$ is shorter than $e$ and the length of any shortest path from $u$ to $v'$ is at most $(1+\delta)|e|$.
\end{definition}

\vspace{\baselineskip}
\begin{algorithm}[H]\label{algo:blocking}
\DontPrintSemicolon
\KwIn{A partially explored graph $G$, and a vertex $y$ of $G$ that is explored for the first time.}
\While{there is an unblocked boundary edge $e=(u,v)$, with $u$ explored and $v$ unexplored, such that $u=y$ or such that $e$ had previously been blocked by some edge $(u',y)$}{
	walk a shortest known path from $y$ to $u$\;
	traverse $e=(u,v)$\;
	$\blocking(G,v)$\;
	walk a shortest known path from $v$ to $y$\;
}
\caption{The exploration algorithm $\blocking(G,y)$ as in \cite{megow}}
\end{algorithm}
\vspace{\baselineskip}

In \cite{kaly} and \cite{megow} the algorithm is considered only for blocking parameters $\delta >0$.
We also examine it for $-1<\delta\leq 0$.
Note that for $\delta \leq -1$ the blocking condition will never be satisfied, implying that $\blocking[-1]$ is simply DFS.

Arguing that the algorithm actually explores the whole graph for the new parameter range works just as for $\delta>0$ in \cite{megow}.
Let $G$ be the graph to be explored.
It is clear that the algorithm terminates since a new vertex is explored in every iteration of the while loop.
Suppose not all vertices have been visited after the termination of the algorithm and let $e=(u,v)$ be a shortest boundary edge at that time. Since no shorter boundary edge exists, $e$ is not blocked.
However, the edge must have been blocked at some point, as otherwise it would have been explored during the call of $\blocking(G,u)$.
Assume $e$ became unblocked for the last time through the traversal of the edge $(x,y)$.
But that means $e$ should have been traversed during the call of $\blocking(G,y)$ which is a contradiction.

On planar graphs, $\blocking$ is $2(2+\delta)(1+2/\delta)$-competitive for $\delta> 0$ and in particular 16-competitive for $\delta=2$ \cite{megow}.
Like in that proof we charge the costs of the algorithms actions to the edges of the explored graph.
Let $B_\delta$ be the cost of $\blocking$, i.e.\ the sum of charges to all edges.
For each iteration of the while loop, the costs of the movements described in the Lines 2, 3 and 5 are charged to the edge traversed in Line 3.
Note that only unblocked boundary edges are charged this way and, in particular, every edge will be charged at most once.
Moreover, the following holds.

\begin{lemma}\label{lem:blocking_charge}
Every edge $e$ that is charged, is charged at most $(4+2\delta)|e|$.
\end{lemma}
\begin{proof}
If the algorithm resides at $u$ before traversing an edge $e=(u,v)$ in Line 3, the action in Line 2 will incur no cost and $e$ will be charged with at most $2|e|$.
Otherwise, if $e$ was previously blocked by an edge $(u',y)$, the blocking condition implies that $d(u,y)\leq (1+\delta)|e|$.
Therefore, the edge $e$ will be charged at most $(1+\delta)|e|$, $|e|$ and $(2+\delta)|e|$ by the movements described in the Lines 2, 3 and 5, respectively.
\end{proof}

\begin{definition}
We call an edge contained in a cycle a \emph{long edge} if it is longer than half the total length of that cycle.
\end{definition}

\begin{lemma}\label{lem:cycles_long_edges}
Let $C$ be a cycle contained in some graph and let $e$ be a long boundary edge on $C$ which is not blocked. 
Then
\begin{equation*}
|e|< \frac{1}{1+\delta}(|C|-|e|).
\end{equation*}
In particular, long boundary edges are always blocked for $\delta> 0$.
\end{lemma}

\begin{proof}
Let $e=(u,v)$. At the time $e$ is a boundary edge, another boundary edge on $C$ exists.
Let $e'=(u',v')$ be the first boundary edge on $C$ encountered when traversing $C-e$ from $u$ to $v$.
Since $e$ is a long edge it is longer than $e'$.
In order for $e$ not to be blocked by $e'$ we must have $(1+\delta)|e|<d(u,v')\leq |C|-|e|$.

For $\delta > 0$, the fact that $e$ is not blocked implies $|e|<|C|-|e|$ which contradicts the definition of a long edge.
\end{proof}

Since a cactus graph only contains edge disjoint cycles, it is easy to see that its optimal tour can be characterized in a similar way to cycles.

\begin{lemma}\label{lem:unicycles_opt}
The optimal tour of a cactus graph traverses all edges which are not contained in a cycle twice. Furthermore, in cycles which contain a long edge all but the long edge are traversed twice while otherwise every edge in the cycle is traversed once.
\end{lemma}

In the following we will prove the previously stated upper bounds on the competitive ratio of $\blocking$ on unicyclic graphs and on cactus graphs.
For unicyclic graphs, we also prove the bound to be tight.
Finally, we prove that for $\delta\leq 0$ the competitive ratio of $\blocking$ on planar graphs is in $\Omega(n)$, i.e. considering the new parameter range for that graph class is not advantageous.

\begin{theorem}\label{thm:unicycles_blocking}
\sloppy On unicyclic graphs, $\blocking$ is $(4+2\delta)$-competitive for $\delta> 0$ and $\max\left(4+2\delta, 3+\frac{\delta^2+\delta/2}{1+\delta} \right)$-competitive for $\delta\leq 0$.
In particular, the algorithm is 3-competitive for $\delta=-\frac{1}{2}$.
\end{theorem}

\begin{proof}
Let $G$ be a unicyclic graph containing a cycle $C$.
By Lemma \ref{lem:blocking_charge}, all edges of $G$ are charged at most $4+2\delta$ times their length.
On the other hand, the optimal tour will traverse all edges in $G$ that are not long at least once according to Lemma \ref{lem:unicycles_opt}.
Since long edges are not charged for $\delta> 0$ by Lemma \ref{lem:cycles_long_edges}, $\blocking$ is $(4+2\delta)$-competitive in this case.

For $\delta\leq 0$ however, a long edge may also be charged.
Should $C$ not contain a long edges or should the long edge in $C$ not be charged, it follows just as for $\delta> 0$ that the competitive ratio is at most $4+2\delta$.

Otherwise, let $e=(u,v)$ be the long edge in $C$ which is traversed from $u$ to $v$ when it is charged.
At the time $e$ is charged there exists a second boundary edge on $C$ which we call $e'=(u',v')$.
Furthermore, let $y$ be the position the algorithm resides at directly before moving to $u$ and traversing $e$.
Let $P$ be the path from $y$ to $C$ which is unique in $G$ (see Figure \ref{fig:unicycles_upper}). Note that $P\subseteq G\setminus C$ and in particular $|P|\leq |G\setminus C|$.
The fact that no edge of $P$ is contained in a cycle is the crucial difference to the case of cactus graphs examined later.

We show that after $\blocking(G,v)$ has been executed all edges on $C$ have been revealed.
Suppose the opposite is true and let $G_0$ be the known subgraph of $G$ after the call of $\blocking(G,v)$ has been executed. Furthermore, let $S$ be the induced subgraph of all vertices in $G_0$ that can be reached from $v$ without traversing $e$.

\input{figures/unicyclic_graphs_upper_article.tex}

Let $e''$ be a shortest boundary edge in $S$.
The situation is shown in Figure \ref{fig:unicycles_upper} where blank nodes represent unvisited nodes and the edges of $S$ are indicated by thick lines.
Since $S$ contains a boundary edge on $C$ which is shorter than $e$, the edge $e''$ must also be shorter than $e$.
That implies that $e''$ cannot be blocked by any boundary edge outside of $S$ because the distance to such an edge is larger than $|e|>(1+\delta)|e|>(1+\delta)|e''|$.
By definition there also exists no smaller boundary edge than $e''$ in $S$ that could block it.
So $e''$ is not blocked after $\blocking(G,v)$ has been executed.
However, this leads to a contradiction since $e''$ should have been traversed during the call $\blocking(G,v)$ either after it was detected if it was not blocked then or otherwise later when it became unblocked.

So since $C$ will be fully revealed when $\blocking$ returns from $v$ to $y$, the algorithm will not traverse $e$ a second time but instead take the shorter way around the cycle.
The charge to $e$ is therefore exactly $|C|+2|P|$.
Since $e$ is a long edge that was traversed, we can apply Lemma \ref{lem:cycles_long_edges} and conclude
\begin{align*}
B_\delta &\leq (4+2\delta)|G \setminus C| + (4+2\delta)(|C|-|e|) + |C|+2|P|\\
&\leq 2(3+\delta)|G \setminus C| + (5+2\delta)(|C|-|e|) + |e|\\
&\leq 2(3+\delta)|G \setminus C| + 2\left(3+\frac{\delta^2+\delta/2}{1+\delta}\right)(|C|-|e|).
\end{align*}
On the other hand, $\mathrm{OPT} = 2|G \setminus C|+2(|C|-|e|)$.
Since $\delta \leq (\delta^2+\delta/2)/(1+\delta)$ for $\delta\leq 0$, we have established the proposed upper bound.
\end{proof}

\begin{theorem}
\sloppy On unicyclic graphs, $\blocking$ has a competitive ratio of at least $(4+2\delta)$ for $\delta > 0$ and of at least $\max\left(4+2\delta, 3+\frac{\delta^2+\delta/2}{1+\delta} \right)$ for $\delta \leq 0$.
\end{theorem}

\begin{proof}
Consider the graph in Figure \ref{fig:unicycles_spiked_path} which we call a \emph{spiked path} $SP_m$ for some $m\in\mathbb{N}$.
Choose $k=\lceil 1+\delta\rceil  +1$ which implies $k>1+\delta$. The graph contains a path from the \emph{entry node} (the leftmost node in the figure) to the \emph{exit node} (the rightmost node). This path is made up of a unit length entry edge followed by $mk$ edges of length $1/k$ and finally $m$ edges $l_1,\ldots,l_m$, where $|l_i|=(i+\sum_{j=1}^{i-1} |l_j|)/(1+\delta)$.
Additionally, the spikes $s_1, \ldots, s_m$, each of length $1/k$, are attached to the path as shown in Figure \ref{fig:unicycles_spiked_path}. The spike $s_1$ is at a distance of $1-1/k$ from $l_1$ while the distance between $s_i$ and $s_{i-1}$ is 1 for $i\geq 2$.

\input{figures/unicyclic_graphs_spiked_path}

The spiked path will later be a subgraph of a larger graph and will only be connected to other parts of the graph by edges incident to the entry or exit node.
Assume that $\blocking$ enters $SP_m$ by traversing the entry edge at a point in time when the exit node is still unseen.
We will prove that the edges $l_1,\ldots,l_m$ are each charged $4+2\delta$ times their length and that their total length dominates the total length of the remaining edges in the graph for large $m$.

Suppose that $\blocking$, after traversing the entry edge, continues along the path up to $l_1$ without exploring any of the spikes on the way.
The length of $l_1$ has been chosen such that the unexplored tip of $s_1$ is now at a distance of exactly $(1+\delta)|l_1|$, implying that $l_1$ is blocked by $s_1$. 
So the algorithm backtracks and traverses $s_1$.
Since this unblocks $l_1$, the algorithm now walks to $l_1$ and traverses it, which together costs $(2+\delta)|l_1|$.
Since $l_2$ is now blocked by $s_2$ the algorithm returns to the tip of $s_1$, which again costs $(2+\delta)|l_1|$.
Next it will traverse $s_2$ which unblocks $l_2$ and the process repeats.
So for $i=1,\ldots,m-1$ the edge $l_i$ is charged exactly $(4+2\delta)|l_i|$.

For $l_m$ this is also true if a shortest path from the exit node to the entry node when considering the whole graph is contained in the subgraph $SP_m$.
This ensures that the return from $l_m$ to the tip of $s_m$ costs $(2+\delta)|l_m|$.
In this case we say that there is \emph{no shortcut outside $SP_m$}.
Otherwise, the edge $l_m$ is nevertheless charged at least $(2+\delta)|l_m|$.

It can be easily verified that no edges in $SP_m$ apart from the entry edge can be blocked by a boundary edge outside $SP_m$.
Finally, the fact that $|l_i|\geq i/(1+\delta)$ for $i=1,\ldots,m$ implies that $\sum_{i=1}^{m} |l_i| \in \Omega(m^2)$.
Meanwhile, the lengths of all other edges of $SP_m$ add up to $1+m(1+1/k) \in \mathcal{O}(m)$.

For both the cases $\delta> 0$ and $\delta \leq 0$ consider the graph made up of two spiked paths $SP_m^{(1)}$ and $SP_m^{(2)}$ which are connected as shown in the left of Figure \ref{fig:unicycles_lower}. The arrows indicate the direction from entry to exit node.

\input{figures/unicyclic_graphs_lower_bound}

Assume $\blocking$ chooses to enter $SP_m^{(1)}$ through its entry edge in its first step. It will then completely explore $SP_m^{(1)}$ before traversing the two unit length edges and entering $SP_m^{(2)}$ through its entry edge. Subsequently, it completely explores $SP_m^{(2)}$.
Hence, $B_\delta \geq 2(4+2\delta)\sum_{i=1}^{m} |l_i|$ since there are not shortcuts outside $SP_m^{(1)}$ or $SP_m^{(2)}$.
On the other hand, $OPT = 2\sum_{i=1}^{m} |l_i| + \mathcal{O}(m)$.
This implies that the ratio $B_\delta / OPT$ comes arbitrarily close to $4+2\delta$ from below when we choose $m$ sufficiently large.

In order to prove the second part of the lower bound for $\delta \leq 0$ consider the spiked path $SP_m$ connected into a cycle as shown in the right of Figure \ref{fig:unicycles_lower}.
Assume that $\blocking$ chooses to traverse the entry edge of $SP_m$ as its first step and then consequently explores the whole of $SP_m$.
Note that in this graph there is a shortcut outside $SP_m$, implying that $l_m$ will only be charged at least $(2+\delta)|l_m|$.
Hence,
\begin{equation*}
B_\delta \geq (4+2\delta)\sum_{i=1}^{m-1} |l_i| + (2+\delta)|l_m| \geq 2\left(3 +\frac{\delta^2 + \delta /2}{1+\delta}\right) \sum_{i=1}^{m-1} |l_i|
\end{equation*}
which follows from $|l_m|>\sum_{i=1}^{m-1} |l_i|/(1+\delta)$.
Finally, $OPT = 2\sum_{i=1}^{m-1} |l_i| + \mathcal{O}(m)$ and $\sum_{i=1}^{m-1} |l_i| \in \Omega(m^2)$ prove the desired lower bound.
\end{proof}

\begin{theorem}\label{thm:cactus_blocking}
On cactus graphs, $\blocking$ is $(4+2\delta)$-competitive for $\delta > 0$ and $\left(4+ \frac{\delta^2 + \delta/2}{1+\delta} \right)$-competitive for $\delta \leq 0$.
In particular, the algorithm is $\frac{5}{2}+\sqrt{2}\approx 3.91$-competitive for $\delta=\frac{1}{\sqrt{2}}-1\approx-0.29$.
\end{theorem}

\begin{proof}
For $\delta> 0$ the proof works just as for unicyclic graphs. 
Let $\delta \leq 0$ and consider a cycle $C$ in a cactus graph.
Assume $C$ contains a long edge $e=(u,v)$ that is traversed from $u$ to $v$ when it is charged.
Analogously to unicyclic graphs we know that after $\blocking(G,v)$ has been executed the cycle $C$ is fully revealed.
It is then possible to traverse $C-e$ instead of $e$ when returning from $v$ to the previously explored vertex.
Hence, the edge $e$ is charged at most $(1+\delta)|e|+|e|+|C|-|e|+(1+\delta)|e|$.
Let $B_\delta(C)$ denote the sum of charges made to edges of $C$. Using Lemma \ref{lem:cycles_long_edges} we conclude
\begin{align*}
B_\delta(C) &\leq (4+2\delta)(|C|-|e|) + (2+2\delta)|e|+|C| \\
&\leq (5+2\delta)(|C|-|e|) + (3+2\delta)|e|\\
&\leq 2\left(4+\frac{\delta^2+\delta/2}{1+\delta}\right)(|C|-|e|).
\end{align*}

Let $e_1,\ldots, e_k$ be the long edges in $G$ which are charged and let $C_1,\ldots, C_k$ be the cycles they are contained in, respectively.
Let $G'$ be the graph $G$ excluding all edges in $C_1,\ldots, C_k$ and excluding all other long edges which are not charged.
Then $\mathrm{OPT}\geq |G'|+\sum_{i=1}^k 2(|C_i|-|e_i|)$ because according to Lemma \ref{lem:unicycles_opt} all edges of the graph which are not long are traversed at least once by the optimal tour.
Furthermore, the bound on $B_\delta(C)$ yields
\begin{equation*}
B_\delta \leq (4+2\delta)|G'|+\left(4+\frac{\delta^2+\delta/2}{1+\delta}\right)\sum_{i=1}^k 2(|C_i|-|e_i|).
\end{equation*}
This proves the proposed competitiveness since $2\delta \leq \frac{\delta^2+\delta/2}{1+\delta}$ for $\delta\leq 0$.

\end{proof}

\begin{theorem}\label{thm:planar_neg}
On planar graphs, the competitive ratio of $\blocking$ is in $\Omega(n)$ for $\delta\leq 0$.
\end{theorem}

\begin{proof}
Let $m\in\mathbb{N}$ and consider the graph in Figure \ref{fig:planar_blocking_linear}.
It contains a path of $m$ unit length edges between the start node $s$ and the node $p$.
Additionally, there are further $m$ paths connecting $s$ and $p$ each of which consists of two unit length edges and one edge of length $m$ in this order.
So the graph contains $n=3m+1$ vertices in total.

\input{figures/planar_graphs_blocking_neg_linear_elsevier.tex}

Suppose $\blocking$ begins its exploration by traversing the path of unit length edges from $s$ to $p$.
When it resides at $p$, all $m$ edges of length $m$ are unblocked since the distance from $p$ to the unexplored vertex of any unit length boundary edges is $m+1 > (1+\delta)m$.
So the algorithm will successively traverse each edge of length $m$ as well as the following unit length edge on the path.
This implies $B_\delta\geq 2m^2$.
On the other hand, the length of the optimal tour is $6m$.
Thus,
\begin{equation*}
\frac{B_\delta}{OPT}\geq \frac{1}{3} m \geq \frac{1}{12}n.
\end{equation*}

\end{proof}

\section*{Acknowledgements}
I would like to thank Prof.\ Dr.\ Susanne Albers for suggesting the topic of and supervising the master's thesis, the main findings of which are presented in this paper.
Furthermore, I would like to thank the anonymous reviewers for their helpful comments and suggestions, in particular, for pointing out a simpler graph construction for the lower bound on trees.

\bibliographystyle{abbrv}
\bibliography{master_article2}

\end{document}

%% file: figures/trees_nn_elsevier.tex
\begin{figure}[htbp]
\centering

\begin{tikzpicture}[scale=1] %
\foreach \x in {0, 3, 7, 10}{
	\draw (\x, 0) -- (\x+1, 0) node[vertex]{};
	\node[vertex] (v\x) at (\x, 0){};
}
\foreach \x in {1, 8}{
    \draw (\x, 0) -- node[weight, swap]{2} (\x+2, 0);
}
\foreach \x in {4}{
    \draw (\x, 0) -- node[weight, swap]{3} (\x+3, 0);
}
\foreach \x in {1, 3, 4, 7, 8, 10, 11}{
	\draw (\x, 0) -- (\x,1) node[vertex]{};
}

\node[weight, below] at (0,0) {$l_3=l_2'$};
\node[weight, below] at (4,0) {$r_2'$};
\node[weight, below] at (7,0) {$l_2''$};
\node[weight, below] at (11,0) {$r_3=r_2''$};

\node[weight, above] at (3,1) {$m_2'$};
\node[weight, above] at (7,1) {$m_3$};
\node[weight, above] at (10,1) {$m_2''$};

\end{tikzpicture}

\caption{The graph $G_3$}
\label{fig:trees_nn_G_k}
\end{figure}
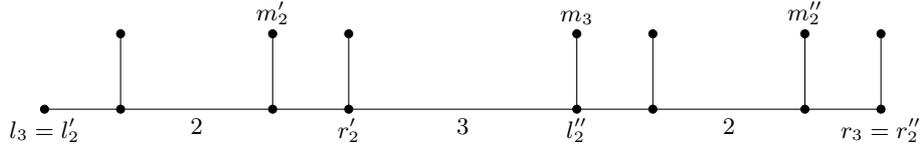

%% file: figures/unicyclic_graphs_upper_article.tex
\begin{figure}[htbp]
\centering

\begin{tikzpicture}[scale=1]

\def\radius{2cm}

\def\drawarc#1#2#3{
	\draw (#1:\radius) node[vertex]{} arc(#1:#2:\radius) node[weight, swap, midway]{#3} node[vertex]{};
}

\def\drawthickarc#1#2{
	\draw[very thick] ++(#1:\radius) arc(#1:#2:\radius);
}

\def\drawdotarc#1#2{
	\begin{scope}[on background layer]
		\draw[continued] ++(#1:\radius) arc(#1:#2:\radius);
	\end{scope}
}

\def\drawpoint#1#2{
	\node[vertex, label={[weight]#1:#2}] at (#1:\radius){};
}

\def\drawboundarypoint#1{
	\node[boundary] at (#1:\radius){};
}

\def\drawpath#1#2#3#4#5{
	\node[vertex] at (#1:\radius){};
	\draw[continued] (#1:\radius) -- +(#1:#2);
	\draw[|-|, very thin] ($(#1:\radius) + (#1-90:#3)$) -- node[weight, swap]{#4} ++(#1:#2);
	\node[vertex, label={[weight]left:#5}] at (#1:\radius+#2){};
}

\def\drawtree#1{
	\draw[very thick] (#1:\radius) coordinate(t0) -- +(#1:1cm) coordinate(t1);
	\node[vertex] at (t0){};
	\node[vertex] at (t1){};
	\draw[very thick] (t1) -- node[weight, swap]{$e''$} +(#1+60:0.5cm) node[](t3){};
	\node[boundary] at (t3){};
	\draw[very thick] (t1) -- +(#1-60:0.8cm) coordinate(t2);
	\node[vertex] at (t2){};	
	\draw[very thick] (t2) -- +(#1-10:0.7cm) coordinate(t4);
	\node[vertex] at (t4){};
	\draw[very thick] (t2) -- +(#1-90:0.7cm) coordinate(t5);
	\node[vertex] at (t5){};

}

\drawboundarypoint{110}

\drawdotarc{110}{180}

\drawpath{150}{3cm}{0.2cm}{$P$}{$y$}

\drawpoint{180}{$u$}

\drawarc{180}{380}{$e$}

\drawpoint{380}{$v$}

\drawthickarc{20}{50}{}

\drawtree{50}

\drawthickarc{50}{75}{}

\drawboundarypoint{75}

\end{tikzpicture}

\caption{Exploring a unicyclic graph}
\label{fig:unicycles_upper}
\end{figure}
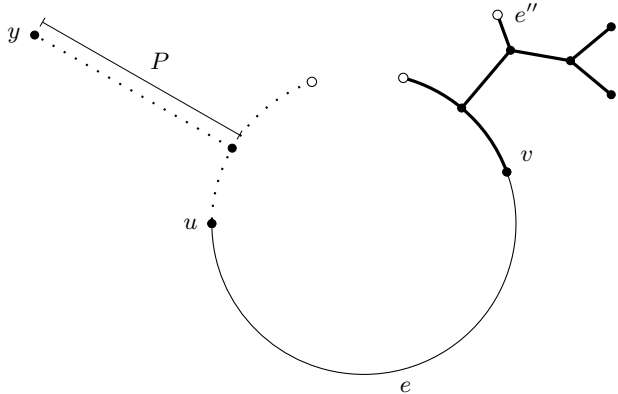

%% file: figures/unicyclic_graphs_spiked_path.tex
\begin{figure}[htbp]
\centering

\begin{tikzpicture}[scale=\textwidth/9cm]

\def\kfrac{0.4cm}

\draw (0,0) node[vertex]{} -- node[weight, swap]{1} +(1,0) node[vertex](v1){};
\draw (v1) -- node[weight, swap]{$\frac{1}{k}$} +(\kfrac,0) node[vertex](v2){};
\draw (v2) -- node[weight, swap]{$s_m$} +(0,\kfrac) node[vertex]{};
\draw[multiple] (v2) -- node[weight, swap]{$1$} +(1,0) coordinate(v3);
\node[vertex] at (v3){};
\draw (v3) -- node[weight, swap]{$s_{m-1}$} +(0,\kfrac) node[vertex]{};

\draw[continued] (v3) -- +(1.2,0) coordinate(v4);
\node[vertex] at (v4){};

\draw (v4) -- node[weight, swap]{$s_2$} +(0,\kfrac) node[vertex]{};
\draw[multiple] (v4) -- node[weight, swap]{$1$} +(1,0) coordinate(v5);
\node[vertex] at (v5){};
\draw (v5) -- node[weight, swap]{$s_1$} +(0,\kfrac) node[vertex]{};
\draw[multiple] (v5) -- node[weight, swap, font=\small]{$1-\frac{1}{k}$} +(0.9,0) coordinate(v6);
\node[vertex] at (v6){};
\draw (v6) node[vertex]{} -- node[weight, swap]{$l_1$} +(1,0) node[vertex](v7){};
\draw[continued] (v7) -- +(0.9,0) coordinate(v8);
\node[vertex] at (v8){};
\draw (v8) -- node[weight, swap]{$l_m$} +(1.4,0) node[vertex]{};
\end{tikzpicture}

\caption{The spiked path $SP_m$ (Dashed lines indicate paths consisting of edges of length $1/k$.)}
\label{fig:unicycles_spiked_path}
\end{figure}
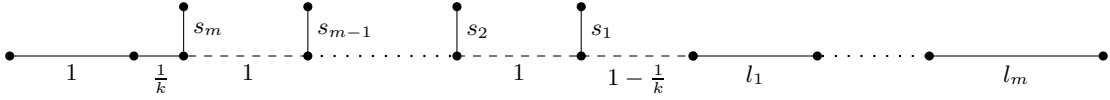

%% file: figures/unicyclic_graphs_lower_bound.tex
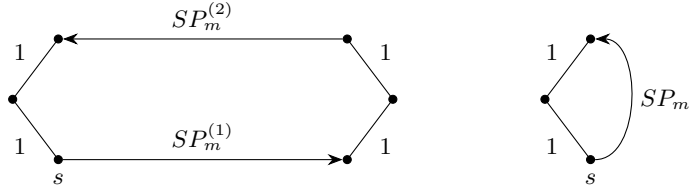
\begin{figure}[htbp]
\centering

\begin{tikzpicture}
\node[vertex] at (0,0) {};
\draw (0,0) -- node[weight]{1} (0.6, 0.8) node[vertex](x2){};
\draw (0,0) -- node[weight, swap]{1} (0.6, -0.8) node[vertex](e1){};
\node[weight, below] at (e1.south) {$s$};

\node[vertex] at (5,0) {};
\draw (5,0) -- node[weight, swap]{1} (4.4, 0.8) node[vertex](e2){};
\draw (5,0) -- node[weight]{1} (4.4, -0.8) node[vertex](x1){};

\draw[arrow] (e1.east) -- node[weight]{$SP_m^{(1)}$} (x1.west);
\draw[arrow] (e2.west) -- node[weight, swap]{$SP_m^{(2)}$} (x2.east);

\node[vertex] at (7,0) {};
\draw (7,0) -- node[weight]{1} (7.6, 0.8) node[vertex](x){};
\draw (7,0) -- node[weight, swap]{1} (7.6, -0.8) node[vertex](e){};
\node[weight, below] at (e.south) {$s$};

\draw[arrow] (e.east) to[bend right=90] node[weight, swap]{$SP_m$} (x.east);
\end{tikzpicture}

\caption{Lower bound constructions for unicyclic graphs}
\label{fig:unicycles_lower}
\end{figure}

%% file: figures/planar_graphs_blocking_neg_linear_elsevier.tex
\begin{figure}[!htbp]
\centering

\begin{tikzpicture}[scale=1]

\node[vertex, label={[weight]below:$s$}] (s) at (0,0){};

\node[vertex, label={[weight]below:$p$}] (p) at (6,0){};

\draw[multiple] (s) -- node[weight, swap]{$m\times 1$} (p);

\foreach \bend in {10, 20, 40}{
	\draw (s) -- (\bend:1) node[vertex]{};
	\draw (\bend:1) -- (\bend:2) node[vertex]{};
	\draw (\bend:2) to[bend left=\bend] (p);
}
\foreach \bend in {50}{
	\draw (s) -- node[weight]{$1$} (\bend:1) node[vertex]{};
	\draw (\bend:1) -- node[weight]{$1$} (\bend:2) node[vertex]{};
	\draw (\bend:2) to[bend left=\bend] node[weight]{$m$} (p);
}

\node at (3,1.3) {\vdots};

\end{tikzpicture}
\caption{Planar lower bound example for $\delta\leq 0$}
\label{fig:planar_blocking_linear}
\end{figure}
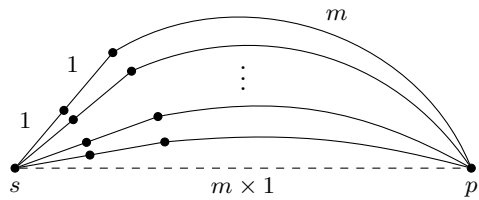